\documentclass{article}

\usepackage{amssymb}
\usepackage{amsmath}
\usepackage{amsthm}
\usepackage{graphicx}
\usepackage{subcaption}
\usepackage[ruled,noend,linesnumbered]{algorithm2e}
\DontPrintSemicolon

\DeclareMathOperator{\E}{\mathbb{E}}
\newtheorem{theorem}{Theorem}
\newtheorem{corollary}[theorem]{Corollary}
\newtheorem{lemma}[theorem]{Lemma}

\newtheorem{claim}[theorem]{Claim}

\newtheorem{example}{Example}

\newcommand{\CLP}{\mathrm{CLP}}

\newcommand{\scoretype}{\mathrm{ScoreIndex}}

\newcommand{\abs}[1]{\lvert #1 \rvert}
\newcommand*{\rev}{\textsc{reverse}{}}

\newcommand{\LPS}{\textsc{LP-solve}}

\newcommand{\Ra}{\mathcal{R}_{\vec{\alpha}}}
\newcommand{\NP}{\mathsf{NP}}

\begin{document}

\title{New Approximations for Coalitional Manipulation in General Scoring Rules}

\date{}
\author{ Orgad Keller\\  \texttt{orgad.keller@gmail.com}\\ 
	  Department of Computer Science, Bar-Ilan University, Israel
	\and
        Avinatan Hassidim\\  \texttt{avinatan@cs.biu.ac.il} \\
        Department of Computer Science, Bar-Ilan University, Israel
       \and
        Noam Hazon\\  \texttt{noamh@ariel.ac.il} \\
         Department of Computer Science, Ariel University, Israel
      }
 
\maketitle

\begin{abstract}
We study the problem of \emph{coalitional manipulation}---where $k$ manipulators try to manipulate an election on $m$ candidates---under general scoring rules, with a focus on the Borda protocol. We do so  both in the weighted and unweighted settings.

For these problems, recent approaches to approximation tried to minimize $k$, the number of manipulators needed to make the preferred candidate win (thus assuming that the number of manipulators is not limited in advance), we focus instead on minimizing the maximum score obtainable by a non-preferred candidate. 

In the strongest, most general setting, we provide an algorithm for any  scoring rule as described by a vector $\vec{\alpha}=(\alpha_1,\ldots,\alpha_m)$: for some $\beta=O(\sqrt{m\log m})$, it obtains an additive approximation equal to $W\cdot \max_i\abs{\alpha_{i+\beta}-\alpha_i}$, where $W$ is the sum of voter weights. In words, this factor is the maximum difference between two scores in $\vec{\alpha}$ that are $\beta$ entries away, multiplied by $W$. The unweighted equivalent is provided as well.

For Borda, both the weighted and unweighted variants are known to be $\NP$-hard. For the unweighted case, our simpler algorithm provides a randomized, additive $O(k \sqrt{m \log m} )$ approximation; in other words, if there exists a strategy enabling the preferred candidate to win by an $\Omega(k \sqrt{m \log m} )$ margin, our method, with high probability, will find a strategy enabling her to win (albeit with a possibly smaller margin). It thus provides a somewhat stronger guarantee compared to the previous methods, which implicitly implied (with respect to the original $k$) a strategy that provides an $\Omega(m)$-additive approximation to the maximum score of a non-preferred candidate: when $k$ is  $o(\sqrt{m/\log m})$, our strategy thus provides a stronger approximation. 

For the   weighted case, our generalized algorithm provides an   $O(W \sqrt{m \log m} )$-additive approximation, where $W$ is the sum of voter weights. This is a clear advantage over previous methods: some of them do not generalize to the weighted case, while others---which approximate the number of manipulators---pose restrictions on the weights of extra manipulators added.

We note that our algorithms for Borda can also be viewed as a $(1+o(1))$-multiplicative approximation since the values we approximate have  natural  $\Omega(km)$ (unweighted) and $\Omega(Wm)$ (weighted) lower bounds.


Our methods are novel and adapt techniques from multiprocessor scheduling 
by carefully rounding an exponentially-large configuration linear program that is solved by using the ellipsoid method with an efficient separation oracle. We believe that such methods could be beneficial in social choice settings as well. 
\end{abstract}


%
%


\section{Introduction}
Elections are one of the pillars of democratic societies, and are  an important part of social choice theory. In addition they have played a major role in multiagent systems, where a group of intelligent agents would like to reach a joint decision~\cite{DBLP:conf/ijcai/EphratiR93}.
In its essence, an election consists of $n$ agents (also called voters) who need to decide on a winning candidate among $m$ candidates. In order to do so, each voter reveals a ranking of the candidates according to his preference and the winner is then decided according to some protocol.

Ideally in voting, we would like the voters to be truthful, that is, that their reported ranking of the candidates will be their true one. 
However, almost all voting rules are prone to manipulation: Gibbard and Satterthwaite~\cite{gibbard1973manipulation,satterthwaite1975strategy} show that for any reasonable preference-based voting system with at least $3$ candidates, voters might benefit from reporting a ranking different than their true one in order to make sure that the candidate they prefer the most wins. Furthermore, several voters might decide to collude, to form a coalition and then to coordinate their votes in such a way that a specific candidate $p$ (hereafter the \emph{preferred candidate}) will prevail. Such a setting is reasonable especially when the voters are agents that are operated by one party of interest.

For some time, the hope for making voting protocols immune to manipulations at least \emph{in practice} relied on computational assumptions: for several common voting protocols, it was shown that computing a successful voting strategy for the manipulators is $\NP$-hard~\cite{bartholdi1989computational,DBLP:journals/jacm/ConitzerSL07,DBLP:conf/ijcai/XiaZPCR09,faliszewski2010ai}. However, as it is many times the case, approximation algorithms and heuristics were devised in order to overcome the $\NP$-hardness albeit with some compromises on the quality of the resulting strategy. This paper fits within this scheme.

In this paper we focus on general scoring rules $\Ra$ and in particular on the Borda voting rule. We first study the problem of (constructive) unweighted coalitional manipulation (\emph{UCM})\footnote{The problem was called CCUM, for ``constructive coalitional unweighted manipulation'', in~\cite{DBLP:journals/ai/ZuckermanPR09}.}: assume that $k$ additional voters (hereafter the manipulators), all of them preferring a specific candidate $p$, can be added to the voting system, thus forming a coalition. Also assume that all $n$ original voters (hereafter the non-manipulators) voted first (or equivalently, that the non-manipulators are truthful and that their preferences are known). Find a strategy for the manipulators telling each one of them how to vote so that $p$ wins, if such strategy exists.  We call such a strategy a \emph{$p$-winning strategy}. In the \emph{weighted} variant (\emph{WCM}), the manipulators are weighted; essentially this means that points awarded by a voter to a candidate are multiplied by the voter's weight.

$\Ra$-WCM, for all positional scoring rules $\Ra$, except plurality-like rules, was shown to be $\NP$-hard when $m\geq 3$~\cite{DBLP:journals/jacm/ConitzerSL07,DBLP:journals/jcss/HemaspaandraH07,DBLP:journals/jair/ProcacciaR07}. Therefore, Borda-WCM is $\NP$-hard. Borda-UCM eluded  researchers for some time; it was first conjectured and finally proven to be $\NP$-hard~\cite{DBLP:conf/aaai/DaviesKNW11,DBLP:conf/ijcai/BetzlerNW11}. 

As a way of overcoming the hardness, recent research~\cite{DBLP:journals/ai/ZuckermanPR09}
focused on an approximation to the minimum number of manipulators needed to be added to the system in order to guarantee that the preferred candidate $p$ would win. For Borda-UCM, they showed that if there exists a $p$-winning strategy for $k$ manipulators, then they will find a $p$-winning strategy with at most one additional manipulator -- besides the $k$ given by the problem definition. For Borda-WCM,  they showed that if there exists a $p$-winning strategy using the $k$ given weighted manipulators, they will find a $p$-winning strategy using additional manipulators, if the sum of weights of the additional manipulators equals the maximum over the weights of the $k$ original manipulators.
 
 This kind of approximation might seem a bit problematic: first, the ability to add a manipulator is a strong operation, perhaps too strong; for instance, for Borda-UCM, adding a manipulator adds $\Omega(m)$  to the difference between $p$ and its highest-scoring competitor. Second, while in some cases it might be reasonable that the party behind the manipulators can add another manipulator to the system, in many cases we do not expect this to be true. Furthermore, in the weighted case assumptions are needed to be made on the weight of the additional manipulators -- also a problematic aspect. Instead, it is interesting to ask what can we assert---assuming that the number of manipulators \emph{cannot be changed}---on the ability to promote a specific candidate $p$ given the non-manipulator scores of all candidates and the value $k$ (or equivalently, the length-$k$ vector of manipulator weights).
 
 We provide a positive result of the following type: 
 \begin{quotation}
\noindent {\bf Main Result:} \emph{If there exists a manipulation strategy enabling $p$ to win by a large-enough margin, we efficiently find a successful manipulation strategy making $p$ win.}
\end{quotation}
  Take the unweighted case as an example: assume that we can provide, for some function $f(k,m)$, an $f(k,m)$-additive approximation to the maximum difference obtainable between $p$'s final score and the final score of the highest-scoring non-preferred candidate. Then, if there exists a $p$-winning strategy such that this difference is at least $f(k,m)$, we can be rest-assured that the algorithm will find a $p$-winning strategy.

 This, in turn, boils down to approximating an upper-bound to the score of the highest-ranked candidate who is not $p$. Earlier research of this flavor focused only on cases where the number of candidates is bounded: for $\Ra$-WCM, Brelsford et al.~\cite{DBLP:conf/aaai/BrelsfordFHSS08} provide an FPTAS to that upper bound (to be exact, they provide an FPTAS to the same exact value  we defined, and then use it to provide another FPTAS to  another value, which is their value-of-interest.\footnote{They are interested in the difference between the score of the preferred candidate and the highest-scoring non-preferred candidate when including the manipulator votes, minus the same difference when not including the manipulator votes. Notice that the  upper-bound we defined is the only  non-trivial value in this computation.}) For $\Ra$-UCM, if the number of candidates is bounded, the entire problem becomes easy and polynomial-time solvable~\cite[Proposition~1]{DBLP:journals/jacm/ConitzerSL07}. Compared to this line of work, we do not limit ourselves to bounded number of candidates.

 \subsection{Our Results and Contributions}
 Consider a general positional scoring rules $\Ra$, as is usually described by a vector $\vec{\alpha} = (\alpha_0,\ldots,\alpha_m)$ (see Section~\ref{sec:def} for full definitions). 
 Now let $T^*=\min_S\max_{c' \in C \setminus \{p\}}\mathrm{score}_{\Ra,E,S}(c')$ be the minimum possible score (ranging over all possible manipulation strategies $S$) of the highest scoring candidate who is not $p$, where $\mathrm{score}_{\Ra,E,S}(c')$ is the final score of a candidate $c'$ (w.r.t.\ voting rule $\Ra$, election $E$, and strategy $S$) and $C$ is the candidate set. 

 Our main technical contribution is a constructive proof to the following two theorems. 
 Let $\beta = d \sqrt{m\log m}$ for some constant $d$, and let $g(\vec{\alpha})=\max_{i=0,\ldots,m-\beta}\abs{\alpha_{i+\beta}-\alpha_{i}}$. In words, $g(\vec{\alpha})$ is the biggest difference between a score in $\vec{\alpha}$ and another score $\beta$ entries away from it.  
 
 \begin{theorem}\label{thr:general}
 	There exists a randomized Monte Carlo algorithm for $\Ra$-UCM which provides a $k\cdot g(\vec{\alpha})$-additive approximation to $T^*$ with an exponentially-small failure probability.
 \end{theorem}  

\begin{theorem}\label{thr:general2}
 	There exists a randomized Monte Carlo algorithm for $\Ra$-WCM which provides a $W\cdot g(\vec{\alpha})$-additive approximation to $T^*$ with an exponentially-small failure probability, where $W$ is the sum of voter weights.
 \end{theorem}  
 
 These theorems immediately pave the way to the following corollaries:
  \begin{corollary}\label{thr:main}
 	There exists a randomized Monte Carlo algorithm for Borda-UCM which provides an $O(k \sqrt{m \log m} )$-additive approximation to $T^*$ with an exponentially-small failure probability.
 \end{corollary}  
 \begin{corollary}\label{thr:main2}
 	There exists a randomized Monte Carlo algorithm for Borda-WCM which provides an $O(W \sqrt{m \log m} )$-additive approximation to $T^*$ with an exponentially-small failure probability.
 \end{corollary}  
 
 Taking Borda-UCM as an example, if there exists a $p$-winning strategy enabling $p$ to win by a margin of $\Omega(k \sqrt{m \log m} )$  compared to the score of the highest-scoring non-preferred candidate, our method will find a $p$-winning strategy (albeit with a possibly smaller margin). Similar guarantees apply in the more general settings.
 
 Notice that for Borda, such approximations can also be seen as a $(1+o(1))$-multiplicative approximation on the  score of the highest-scoring non-preferred candidate, and thus is superior to an FPTAS; to see that, notice that for Borda-WCM the overall `voting mass' given by the manipulators is $\Omega(W m^2)$, and so the highest scoring candidate has score of at least $\Omega(W m)$. Therefore  $\tilde{O}(W\sqrt{m})$ \footnote{The $\tilde{O}$ notation suppresses poly-logarithmic factors.} is a lower order term.
 This is an advantage over the previous methods: 
 \begin{itemize}
 	
 	\item  Opposed to the heuristics in~\cite{DBLP:journals/ai/DaviesKNWX14}, we provide \emph{provable} guarantees.
 	\item Also opposed to the heuristics in~\cite{DBLP:journals/ai/DaviesKNWX14}, our algorithm generalizes to the weighted case. 
 	\item Compared to the \rev{} algorithm of~\cite{DBLP:journals/ai/ZuckermanPR09} for Borda-UCM, while adding only a single extra manipulator sounds like a minor operation, it is not; as mentioned, an extra manipulator implies the addition of $\Omega(m)$ points to the difference between $p$ and its highest-scoring competitor. 	
 	\item Consider the \rev{} algorithm of~\cite{DBLP:journals/ai/ZuckermanPR09},	
and assume that adding extra manipulators is not allowed. We will show that their method implies no better than an $\Omega(m)$-additive approximation to the score of the highest-scoring non-preferred candidate. Our approximation is thus
 	superior when $k$ is $o(\sqrt{m/\log m})$.
 	\item Compared to previous methods, are results are linear-programming-based, and not  greedy. Thus, they make a decision based on the entire input, as opposed to repeatedly making a decision based on a greedy estimate w.r.t.\ some subset of the problem.
 	 \end{itemize}
 	 
The following claim analyzes \rev{} according to our metric. It is proven in Section~\ref{sec:lower}.
\begin{claim}\label{thr:lower}
	For any $m$, when the addition of more than $k$ extra manipulators is not allowed, there are families of cases in which the optimal strategy enables $p$ to win  by a margin of at least $m/3$, but \rev{} fails to find a $p$-winning strategy.
\end{claim}

 Our techniques are novel: they employ the use of \emph{configuration linear programs} (C-LP), a method that is also used in the scheduling literature, namely for two well-studied problems, the problem of \emph{scheduling on unrelated machines}~\cite{DBLP:journals/siamcomp/Svensson12}, and the so-called \emph{Santa Claus problem}~\cite{DBLP:conf/stoc/BansalS06}. See Section~\ref{subsec:related} for a discussion of these problems. It is important to note that the solutions to the two above problems and to ours all differ from one another with respect to how the algorithm proceeds once the C-LP result is computed.

 C-LPs are used 
  for the generation of an initial, invalid strategy, which is later modified to become valid. They are unique in the sense that they are linear programs that have an exponential number of variables, an issue which we solve by referring to the LP dual and using the ellipsoid method with a polynomially-computable separation oracle~\cite{khachiyan1980polynomial,DBLP:journals/combinatorica/GrotschelLS81}. We have also implemented our algorithm: as a result of not finding a library which enables solving an LP this way, we simulated this by an iterative use of a general LP-solving library, each time adding a violated constraint based on running the separation oracle externally.


%
%
\subsection{Related Work}\label{subsec:related}
\paragraph{Borda.} The Borda voting mechanism was introduced by Jean-Charles de Borda in 1770. It is used, sometimes with some modifications, by parliaments in countries such as Slovenia, and competitions such as the Eurovision song contest, selecting the MVP in  major league baseball, Robocup robot soccer competitions and others.
The Borda voting mechanism is described as follows:  every agent ranks the candidates from $1$ to $m$, and awards the candidate ranked $i$-th a score of $m-i$. Notice that this makes the scores given by each single voter a permutation of $0,\ldots,m-1$. Finally, the winning candidate is the one with the highest aggregate score. 

\paragraph{Easiness Results.} The computational complexity of  coalitional manipulation problems was studied extensively. For general scoring rules $\Ra$, most earlier work considered the case where the number of candidates is bounded: Conitzer at al.~\cite{DBLP:journals/jacm/ConitzerSL07} show that when $m$ is bounded, $\Ra$-UCM is solvable in polynomial time.

Even when $m$ is unbounded, Plurality-UCM and Veto-UCM are still easy using the greedy algorithm of Zuckerman et al.~\cite{DBLP:journals/ai/ZuckermanPR09}. This also holds for $t$-approval-UCM, which generalizes both~\cite{lin2012solving}.

\paragraph{$\NP$-Hardness Results.} In the weighted case, the situation is different: for all positional scoring rules $\Ra$, except plurality-like rules, $\Ra$-WCM is $\NP$-hard when $m\geq 3$~\cite{DBLP:journals/jacm/ConitzerSL07,DBLP:journals/jcss/HemaspaandraH07,DBLP:journals/jair/ProcacciaR07}. In particular, this holds for Borda-WCM.
	
However, the computational hardness of Borda-UCM still remained open for quite some time, until finally shown to be $\NP$-hard as well~\cite{DBLP:conf/aaai/DaviesKNW11,DBLP:conf/ijcai/BetzlerNW11} in 2011, even for the case of $n=3$ and adding $2$ manipulators.

\paragraph{Approximating the number of manipulators.} Zuckerman et al.~\cite{DBLP:journals/ai/ZuckermanPR09} present a greedy algorithm later referred to as \rev{}\footnote{This name was given in~\cite{DBLP:conf/aaai/DaviesKNW11}.}. \rev{} works as follows: after the non-manipulators had finished voting, we go over the manipulators one by one, and each manipulator will rank the candidates (besides $p$) by the \emph{reversed} order of their aggregated score so far (candidate with the highest score so far gets the lowest ranking). As mentioned, for Borda-UCM, \rev{} can be seen as an additive $+1$ approximation for the objective of finding the minimum number of manipulators needed. 

For Borda-WCM, their approximation has the following flavor. Let $S$ be the set of the $k$ given weighted manipulators. If there exists a $p$-winning strategy using $S$, they will find a $p$-winning strategy using additional manipulators, if the sum of weights of the additional manipulators equals $\max_{\ell \in S} w_\ell$, where $w_\ell$ is the weight of manipulator $\ell \in S$.

\paragraph{Approximating the Maximum Score of a Non-Preferred Candidate.}
Returning to earlier results, when $m$ is bounded, Brelsford et al.~\cite[Lemma 3]{DBLP:conf/aaai/BrelsfordFHSS08} provide an FPTAS with respect to the maximum score of a non-preferred candidate. As mentioned, this paves the way for an FPTAS on their value-of-interest.

\paragraph{Heuristics for Borda.}
 Davies et al.~\cite{DBLP:journals/ai/DaviesKNWX14} present two additional heuristics: iteratively, assign the largest un-allocated score to the candidate with the largest gap (\textsc{Largest Fit}), or to the candidate with the largest ratio of gap divided by the number of scores yet-to-be-allocated to this candidate (\textsc{Average Fit}). To the best of our knowledge, these algorithms do not have a counterpart for the weighted case.

\paragraph{Configuration Linear Programs.}
As discussed, configuration linear programs were also used in scheduling literature, for example for the following two problems which were extensively studied before:
\begin{itemize}
	\item In the so-called \emph{Santa Claus problem}~\cite{DBLP:conf/stoc/BansalS06}, Santa Claus has $t$ presents that he wishes to distribute between $m$ kids, and $p_{i,j}$ is the value that kid $i$ has to present $j$. The goal is to \emph{maximize} the happiness of the least happy kid: $\min_{i}\sum_{j\in S_i} p_{i,j}$, where $S_i$ is the presents allocated to kid $i$.
	\item  In the problem of \emph{scheduling on unrelated machines}~\cite{DBLP:journals/siamcomp/Svensson12}. We need to assign $t$ jobs between $m$ machines, and $p_{i,j}$ is the time required for machine $i$ to execute job $j$. The goal is to \emph{minimize} the makespan $\max_{i}\sum_{j\in S_i} p_{i,j}$, where $S_i$ is the jobs assigned to machine $i$.
\end{itemize}
Both papers researched a natural and well-researched `restricted assignment' variant of the two problems where $p_{i,j}\in \{p_j, 0\}$.  In~\cite{DBLP:conf/stoc/BansalS06}, they obtained an $O(\log \log m / \log\log\log m)$-multiplicative approximation to the first problem and in~\cite{DBLP:journals/siamcomp/Svensson12}, they obtained a $(33/17+\epsilon)$-multiplicative approximation to the second.

\section{Preliminaries}
\subsection{Problem Definition}\label{sec:def}
\paragraph{Candidate Set.} With a slight change of notation, let $C=\{c_0,c_1,\ldots,c_m\}$ be a candidate set consisting of  the preferred candidate $p=c_0$ and the other $m$ candidates $c_1,\ldots,c_m$. Note that we changed the notation so that the overall number of candidates is $m+1$; this will help streamline the writing.

\paragraph{Election.} An election $E=(C,V)$ is defined by a candidate set $C$ and a set of voters $V$ where each voter submits a ranking of the candidates according to its preference. Formally, we define $V=\{\succ_1,\ldots,\succ_n\}$, where each $\succ_\ell \in V$ is a total order of the candidates. For example, $c_1 \succ_\ell p \succ_\ell c_2$ is one such possible order if $C=\{p,c_1,c_2\}$. Then, some \emph{decision rule} $\mathcal{R}$ is applied in order to decide on the winner(s); formally $\mathcal{R}(E) \subseteq C$ is the set of winners of the elections. In the specific case of a \emph{positional scoring rule} $\Ra$, the rule is described by a vector  $\vec{\alpha} = (\alpha_0,\alpha_1,\ldots,\alpha_m)$ for which $\alpha_0\leq \alpha_1 \leq \cdots \leq\alpha_m$, and $\alpha_m$ is polynomial in $m$, used as follows: each voter awards $\alpha_i$ to the candidate ranked $(m-i)$-th\footnote{Usually $\vec{\alpha}$ is defined in a descending manner, such that $\alpha_0\geq \alpha_1 \geq \cdots \geq\alpha_m$, and $\alpha_i$ is awarded to the candidate ranked $i$-th. Our choice helps streamline the presentation.}. Finally, the winning candidate is the one with the highest aggregated score. In the specific case of Borda scoring rule, we have that $\vec{\alpha}=(0,1,\ldots,m-1,m)$.

\paragraph{WCM and UCM.} In the $\Ra$-(constructive) weighted coalitional manipulation ($\Ra$-WCM) problem, we are given as input:
\begin{itemize}
	\item  A score profile vector $(\sigma_0, \sigma_1,\ldots,\sigma_m)$ representing the aggregated scores given so far to each candidate in $C$ by the original voters in an election $E$ under the rule $\Ra$. Notice that $(\sigma_0, \sigma_1,\ldots,\sigma_m)$ eliminates the need for obtaining $E$ as input: as we have no control on the truthful voters $V$,  $(\sigma_0, \sigma_1,\ldots,\sigma_m)$ is thus a sufficient representation for the outcome of non-manipulator votes.
	\item  A vector $\vec{w}=(w_1,\ldots,w_k)$ of positive integers representing the weights of $k$ manipulators who will be added to the election. The weights have the following meaning: each manipulator $\ell=1,\ldots,k$ is replaced by $w_\ell$ identical but unweighted copies of himself. 
\end{itemize}

It then should be determined if when adding the $k$ manipulators then either (a) no strategy under $\Ra$  exists in which $p$ wins, or that (b) there exists a voting strategy under $\Ra$ such that $p$ can win. In this case, the algorithm should find it. $\Ra$-(constructive) unweighted coalitional manipulation ($\Ra$-UCM) is the specific case where $\vec{w}$ is the all-ones vector and therefore can be replaced in the input by the integer $k$. 


\paragraph{Manipulation Matrices.} Note that in case (b), the output is a voting strategy $S$ which  can be represented as a $k\times (m+1)$ matrix in which the entry $S_{\ell,i}$ describes the score given by manipulator $\ell$ to candidate $c_i$, and where each row of $S$ is a permutation of $\vec{\alpha}$. Such a representation is also called a \emph{manipulation matrix}. 
We can relax the 
requirement that each row of $S$ is a permutation, and replace it by the requirement that each \emph{score-type} $\alpha_j$, is repeated exactly $k$ times in $S$. Such a matrix is called a \emph{relaxed manipulation matrix}. We can perform this relaxation as Davies et al.~\cite[Theorem~7]{DBLP:journals/ai/DaviesKNWX14} show that each relaxed manipulation matrix can be rearranged to become a valid manipulation matrix while preserving each candidate's final score. 

\paragraph{High Probability.} Throughout the paper, when we use the term `with high probability', we mean an arbitrarily-chosen polynomially-small failure probability, i.e., success probability of the form $1-m^{-d}$ where $d \geq 1$ is a constant that can be chosen without affecting the asymptotic running time. `Failure' refers to the event that the algorithm does not provide the desired approximation guarantee.

In this paper we will use various forms of the Hoeffding inequalities, which are variants of the Chernoff inequalities:
\paragraph{Generalized Hoeffding inequality~\cite[Theorem~2]{hoeffding1963probability}.}
Let $X_1,\ldots,X_m$ be independent random variables where each $X_i$ is bounded by the interval $[a_i,b_i]$ respectively. Let $X=\sum_{i=1}^m X_i$ and $\bar{X}=X/m$. Then
\begin{equation*}
\Pr[\bar{X}-\mathbb{E}[\bar{X}] \geq t ] \leq \exp\left(-\frac{2 n^2t^2}{\sum_{i=1}^m (b_i - a_i)^2}\right)\ .
\end{equation*}
By redefining the above inequality in terms of the sum $X$ (instead of the mean $\bar{X}$) and defining $\lambda=tm$ we derive the following equivalent formulation:
\begin{equation*}
\Pr[X-\mathbb{E}[X] \geq \lambda ] \leq \exp\left(-\frac{2 \lambda^2}{\sum_{i=1}^m (b_i - a_i)^2}\right)\ .
\end{equation*}
Specifically, when $X_i \in \{0,1\}$, we obtain the following `classic' Hoeffding inequality, which is equivalent to~\cite[Theorem~1]{hoeffding1963probability}:
\begin{equation*}
\Pr[X-\mathbb{E}[X] \geq \lambda ] \leq \exp\left(-\frac{2 \lambda^2}{m}\right)\ .
\end{equation*}

\subsection{Reduction to a Pure Min-Max Problem}
Since we know what will be the final score of $p$ (non-manipulator votes are known and each manipulator will give $p$ the maximum score possible), we can effectively discard $p$ and treat the problem as a minimization problem on the final scores of $c_1,\ldots,c_m$ only. In other words,  we focus on finding $\min_S\max_{c' \in C \setminus \{p\}}s(c')=\min_S\max_{i=1,\ldots,m}s(c_i)$, where $s(c')$ is $c'$'s final score. Thus the output $S$ is actually a $k \times m$ relaxed manipulation matrix.

Another thing to note is that we can not assume anything about the values in the initial score profile $(\sigma_0, \sigma_1,\ldots,\sigma_m)$; this follows from~\cite[Lemma~1]{DBLP:journals/ai/DaviesKNWX14}, where it is shown that in the context of Borda, for any given non-negative integer vector $(\sigma_0,\sigma_1,\ldots,\sigma_m)$, we can define a set of non-manipulators (along with their preferences) and an additional candidate $c_{m+1}$ that will induce an initial score profile $(T'+\sigma_1,\ldots,T'+\sigma_m,y)$ for some values $T'$ and $y < T'$. Since such an additive translation and the addition of a candidate that will be awarded less than any other candidate have no influence on the winner (and on the difference between each two candidates' final scores), and since our results are concerned with an \emph{additive} approximation, we should assume no prior limitation on the nature of values in $(\sigma_0, \sigma_1,\ldots,\sigma_m)$.

\section{Lower Bound for REVERSE}\label{sec:lower}
We start by showing that there are cases in which the \rev{} algorithm for Borda gives only an $\Omega(m)$-additive approximation to minimum final score of the highest-scoring non-preferred candidate.

\begin{proof}[Proof of Claim~\ref{thr:lower}]
	We provide an infinite family of cases where the claim holds. 
	
	Let $k=3$ and let $m=3t$ for some integer $t$. Consider the case where after the non-manipulators voted, all candidates (but $p$) have the same score $\sigma_i=s$ for all $i$. Effectively this can be normalized to $(\sigma_1,\ldots,\sigma_m)=\vec{0}$.
	
	By the  \rev{} algorithm, the first manipulator can award $c_1,\ldots,c_m$ with $0,\ldots,m-1$ respectively, after which the second manipulator will be obliged to award $c_1,\ldots,c_m$ with $m-1,\ldots,0$ respectively. Repeat this process with the rest of the manipulators, until the final one. It can be verified that $c_m$ will end up with the maximal score of $\lceil k/2 \rceil (m-1)=2(m-1)$.
	
	Conversely, as an upper bound for an optimal solution, consider the following strategy: place all scores to be given in a descending sequence, that is the sequence $\langle m-1,m-1,m-1, m-2,m-2,m-2,\ldots,0,0,0\rangle$. Give the first $m$ scores in the sequence to $c_1,\ldots,c_m$ respectively, the next $m$ to $c_m,\ldots,c_1$ respectively, and the last $m$ to $c_1,\ldots,c_m$ respectively. Since every score-type has $3$ copies, we have just described a relaxed manipulation matrix and therefore by Davies et al.~\cite[Theorem~7]{DBLP:journals/ai/DaviesKNWX14} it can be rearranged to become a valid manipulation matrix without changing the final score of each candidate. Now notice that the score given to any candidate is of the form $(m-r) + (m/3+r-1) + (m/3 -r)=5m/3-r-1$ for some $r \in \{1,\ldots,m/3\}$. As this is at most $5m/3-2$ (when $r=1$), the difference is thus $m/3=\Omega(m)$.
	%
\end{proof}

\section{Linear Programming for UCM}\label{sec:lp-U}
We will begin by providing a ``natural'' way to formulate the min-max version of the problem as an Integer Program (IP). As solving IPs is $\NP$-hard, we will relax it to the equivalent Linear Program (LP). However, such a natural LP will not be useful in our setting, and we will thus introduce a totally different LP formulation, called \emph{Configuration Linear Programming} (C-LP). The number of variables in the C-LP is exponential in the size of the input. Nevertheless, we show that our C-LP can be solved in polynomial time.

Let $[m]=\{1,\ldots,m\}$ and $[m]_0=\{0,\ldots,m-1\}$. We define the variables $x_{i,j}$ for $(i,j) \in [m]\times[m]_0$, and the variable $T$, with the intent that $x_{i,j}$ will equal the number of times candidate $c_i$ received a score of $\alpha_j$, and $T$ will serve as the upper-bound on each candidate final aggregate score. The IP can then be stated as follows:
 \begin{equation*}
 	\min_{\vec{x}} T
 \end{equation*} 
 subject to:
\begin{align}
 &\sum_{i=1}^m x_{i,j}=k &\forall j\in[m]_0\label{eq:orig:1}\ ,\\
 &\sum_{j=0}^{m-1}x_{i,j}=k &\forall i\in[m]\label{eq:orig:2}\ ,\\
 &\sum_{j=0}^{m-1}\alpha_j x_{i,j}\leq T-\sigma_i  &\forall i\in[m]\label{eq:orig:3}\ ,\\
&x_{i,j} \in \{0,\ldots,k\} & \forall i\in[m],j\in[m]_0\ ,
 \end{align}

 where (\ref{eq:orig:1}) guarantees that every score was awarded $k$ times, (\ref{eq:orig:2}) guarantees that every candidate was given $k$ scores, and (\ref{eq:orig:3}) guarantees that  every candidate gets at most $T$ points.

It should be noted that when treating the problem as a min-max problem, we need to take $T$ as a variable that we wish to minimize (this is done by the objective function). However, if we consider the original definition in which our aim is to make the preferred candidate $p$ win, $T$ can be set to $\sigma_0 + k\alpha_{m}$ (the final score of the preferred candidate), and the IP will not have an objective function.

\subsection{Integrality Gap of the Natural LP}
  While we can relax this IP into an LP by replacing the set in the last constraint to be the continuous interval $[0,k]$, it will not be as helpful. We will shortly show why; however note that as this sub-section relates to the deficiencies of the original LP formulation, it can be safely skipped and is not needed for the full understanding of our algorithms.

Consider a ``pure'' LP rounding algorithm, applied w.l.o.g.\ to a minimization problem. Such an algorithm works as follows: given an instance of a problem, it solves its associated relaxed LP (the problem's natural IP where integrality constraints are replaced with their continuous counterparts) and then rounds the resulting solution in some way or the other such that a valid (non-necessarily optimal) solution to the original IP is obtained. The approximation analysis of such algorithms is based on reasoning about how worse is the objective value of the rounded solution compared to the fractional one. In other words, what is the increase---or ``damage done''---to the optimum objective incurred by the rounding process. 
 Since the fractional optimum of a relaxed LP is a lower bound to the integral optimum, i.e., the optimum of the original problem, the same factor also upper-bounds the difference between the objective value of the rounded solution and the one of the integral optimum. Thus this process derives an approximation guarantee. We show that in our case, the increase can be $\Omega(m)$, by showing that there are cases in which the difference between the integral and fractional optimum objective values  is $\Omega(m)$, and thus an algorithm solely based on the rounding procedure cannot hope for an $o(m)$ additive approximation. This kind of reasoning is known as an \emph{integrality gap}, and is demonstrated by the following:
  \begin{lemma}
  	For Borda-UCM, an algorithm based solely on rounding the relaxed natural LP cannot obtain $o(m)$ additive approximation.
  \end{lemma}
\begin{proof}
	We show a lower bound on the approximation ratio in the form of an additive integrality gap. In other words, we show an infinite family of instances where the integral solution to the LP (and thus, to the original problem) gives $\Omega(m)$ worse objective value when compared to  the fractional solution.
	
	Consider the simple case of $m$ candidates, all having equal initial score (w.l.o.g.\ $\sigma_i=0$ for all $i=1,\ldots,m$) and a single manipulator. When solving the problem, one candidate will be awarded $m-1$ and thus will have final score of $m-1$. However, in the fractional solution, the optimum is obtained by splitting each score equally, that is, setting $x_{i,j}=1/m$ for every $i \in [m]$ and $j \in [m]_0$. Now every candidate obtained a final score of $1/m \cdot \sum_{j=0}^{m-1}j=(m-1)/2$. Therefore notice that the gap between the objective of the integral and fractional solutions is $(m-1)/2=\Omega(m)$. 
\end{proof}

\subsection{Introducing Configuration LPs}
In order to work around this we will have to resort to a totally different approach, in which variables no longer represent score types, and instead represent the set of scores (\emph{configuration}) that can be awarded to a candidate.
  
Formally, a \emph{configuration}
$C$ for some candidate $c_i$ is a vector of dimension $m$ in which $C_j$ represents a number of scores of type $\alpha_j$ that $i$ has received, and for which $\sum_{j=0}^{m-1} C_j=k$, that is, the overall number of scores awarded is $k$. For a candidate $c_i$ and a bound $T$, let $\mathcal{C}_i(T)$ be the set of  configurations that do not cause the candidate overall score to surpass $T$, i.e., the set of configurations $C$ for which  $\sum_{j=0}^{m-1} C_{j} \alpha_j\leq T-\sigma_i$.
  
    We formulate the configuration LP as follows:
    \begin{align}
    &\sum_{C \in \mathcal{C}_i(T)}x_{i,C}\leq 1 &\forall i\in[m]\ ,\label{eq:U1}\\
    &\sum_{\substack{i,C\\C \in \mathcal{C}_i(T)}} C_{j} x_{i,C}\geq k &\forall j\in[m]_0\ ,\label{eq:U2}\\
    &x_{i,C} \geq 0 & \forall i\in [m], C \in \mathcal{C}_i(T)\ .
    \end{align}     
  where we wish that the $x_{i,C}$'s would serve as indicator variables indicating whether or not $c_i$ was awarded with configuration $C$, (\ref{eq:U1}) guarantees that every candidate was given at most $1$ configuration and (\ref{eq:U2}) guarantees that every score was awarded at least $k$ times. The choice of inequalities over equalities will be explained soon.

  \begin{example}
  	Consider the case where $k=2$, $m=5$, $\vec{\alpha}=(0,1,2,3,4)$ (i.e., Borda) and  $(\sigma_1,\ldots,\sigma_5)=(5,6,6,6,7)$. We are omitting the non-manipulator votes that provided $\vec{\sigma}$, however recall that there is a possible non-manipulator voting yielding any $\vec{\sigma}$ up to an additive factor and an addition of a candidate.
 Now assume $T=10$ (this is indeed the optimum). Let us focus on the last candidate $c_5$. $\mathcal{C}_5(T)$ should therefore contain all configurations which award $c_5$ at most $T-\sigma_5=3$ points. Those configurations are $(2,0,0,0,0)$ ($0$ points), $(1,1,0,0,0)$ ($1$ point), $(0,2,0,0,0)$, $(1,0,1,0,0)$ ($2$ points), and $(1,0,0,1,0)$ ($3$ points).
  	
  	When solving the $C-LP$, only two of her configurations will get non-zero value: $x_{5,(1,0,0,1,0)} \approxeq 0.7$ and $x_{5,(0,1,1,0,0)} \approxeq 0.3$. We omit the variables corresponding to the rest of the candidates.
  \end{example}
  
  
After solving the LP, we will execute a rounding procedure that will transform the fractional LP solution into a valid solution for the original problem. This procedure can increase the score of some of the candidates, and hence we wish to start with the smallest possible $T$ (so that even after the increase the final score will hopefully be bounded by $\sigma_0 + k\alpha_{m}$).
  
  To find the smallest possible $T$, we perform a one-sided binary search on the value of $T$. For this purpose, for each possible value of $T$ that we come across during the binary search, we redefine the LP and then solve the new LP from scratch, and see if it has a valid solution.    
    The reason we do not add $T$ as a variable in an objective function (instead of the binary search) is that the number of summands in Equations~(\ref{eq:U1},\ref{eq:U2}) depends on $T$.
  
%

 
 This formulation has the obvious drawback that the number of variables is exponential in $k$. However, following the approach of~\cite{DBLP:conf/stoc/BansalS06}, if we find a polynomially-computable separation oracle we can solve the LP by referring to the LP dual and using
 the ellipsoid method. Such an oracle will require a solution to the following seemingly unrelated problem as a subroutine: a variant of the classic Knapsack problem.
 \subsection{$k$-Multiset Knapsack} Let $\{1,\ldots,m\}$ be a set of distinct items, where each item has an associated value $v_j$ and a weight $w_j$. We also obtain a weight upper-bound $W$ and a value lower-bound $V$. As opposed to ordinary knapsack, we also obtain an integer $k$. We are required to find a multiset $S$ of exactly $k$ items (i.e., we can repeat items from the item-set), such that $S$'s overall weight is at most $W$ and $S$'s overall value is greater than $V$ (or to return that no such multiset exists).
 
 \begin{lemma}\label{lem:multiknap}
 	The $k$-multiset knapsack can be solved in time polynomial in $k$, $m$, and $W$ (which is pseudo-polynomial due to the dependence on $W$).
 \end{lemma} 
 \begin{proof}
 	We fill out a table $Q[w,\ell]$, for $w=0,\ldots,W$ and $\ell=0,\ldots,k$, in which $Q[w,\ell]$ is the highest value obtainable with a size-$\ell$ multiset of items of aggregate-weight at most $w$. Notice that $Q$ can be filled using the following recursion: 
 	\begin{equation}
 	Q[w,\ell]=
 	\begin{cases}
 	0 & \text{if $\ell = 0$, }\\
 	\max_j v_j  +  Q'(w-w_j,\ell-1)  & \text{otherwise,}
 	\end{cases} 
 	\end{equation} 
 	where $Q'(w,\ell)=Q[w,\ell]$ if it is defined, i.e., $w\geq 0$ and $\ell \geq 0$, and otherwise is $-\infty$.
 	
 	Therefore $Q$ can be filled-out using dynamic programming. Finally, the entry $Q[W,k]$ contains the highest value obtainable with overall weight at most $W$. Therefore, if $Q[W,k] > V$, we have found a required multiset; otherwise such does not exists. The resulting multiset itself can be recovered using backtracking on the table $Q$. Noticed that the amount of work done is $O(Wkm)$. 
 \end{proof}

\subsection{Solving the UCM C-LP}
 We return to our problem. The choice of inequalities over equalities is motivated by our use of the LP dual in Theorem~\ref{thr:poly-U}. However, they have the same effect as equalities, as shown by the following lemma:
  \begin{lemma}
  	In a solution to the above C-LP, Equations~(\ref{eq:U1},\ref{eq:U2}) will actually be equalities.
  \end{lemma}
  \begin{proof}
  	Notice that by Equation~(\ref{eq:U2}):
  	\begin{align}
  	km &\leq \sum_{j=0}^{m-1} \sum_{i=1}^m \sum_{C \in \mathcal{C}_i(T)} C_{j} x_{i,C} \\
  	&=\sum_{i=1}^m \sum_{C \in \mathcal{C}_i(T)} x_{i,C}\sum_{j=0}^{m-1} C_j\\
  	&=k\sum_{i=1}^m \sum_{C \in \mathcal{C}_i(T)} x_{i,C}\label{eq:U3}\\
  	&\leq km \label{eq:U4}
  	\end{align}
  	where~(\ref{eq:U4}) holds by plugging (\ref{eq:U1}) into (\ref{eq:U3}). We therefore obtain  that $$\sum_{j=0}^{m-1} \sum_{i=1}^m \sum_{C \in \mathcal{C}_i(T)} C_{j} x_{i,C}=km$$ which forces both above non-trivial LP inequalities to be equalities.
  \end{proof}
 
 \begin{theorem}\label{thr:poly-U}
 	Given a fixed value $T$, the UCM C-LP can be solved in polynomial time.
 \end{theorem}
 \begin{proof}
 	We need to refer to the LP dual for our C-LP in order to solve it; we briefly repeat some  LP duality concepts here, refer to~\cite{schrijver1998theory} for complete definitions and discussion.
 	 
 	 The dual of a maximization problem is a minimization problem. In order to define it we can treat our primal program as a maximization problem having all coefficients $0$ in its objective function. 
 	In the dual there is a variable for every constraint of the primal, and a constraint for every variable of the primal. Therefore, we define a variable $y_i$ for each candidate $c_i$ and a variable $z_j$ for each score-type $\alpha_j$ (since the primal has a constraint for each candidate $c_i$ and for each score-type $\alpha_j$). However, since our primal has an exponential number of variables, the dual will have an exponential number of constraints. We will show how to address this.
 	
 	In short, the non-trivial constraints are then obtained by transposing the constraint-coefficient matrix of the primal, using the primal objective function coefficients as the right-hand side of the dual constraints, and using right-hand side of the primal constraints as the coefficients of the dual objective function.

 	The process yields the following dual:
 	\begin{align*}
 	&\min_{\vec{y},\vec{z}} \sum_{i=1}^m y_i - k\sum_{j=0}^{m-1} z_j \\
 	\mbox{subject to:}\\
 	\sum_{j=0}^{m-1} C_{j}z_j&\leq y_i && \forall i\in [m], C \in \mathcal{C}_i(T) \\
 	y_i &\geq 0 && \forall i=1,\ldots,m\\
 	z_j &\geq 0 && \forall j=0,\ldots,m-1
 	\end{align*}
 	As mentioned, the dual has an exponential number of constraints. However it is  solvable; the \emph{ellipsoid method}~\cite{khachiyan1980polynomial} is a method for solving an LP which iteratively tries to find a point inside the feasible region described by the constraints. However, we do not need to provide all the constraints in advance. Instead, the algorithm can be provided with a subroutine, called a \emph{separation oracle}, to which it calls with a proposed point, and the subroutine then either confirms that the point is inside the feasible region or that it returns a violated constraint~\cite{DBLP:journals/combinatorica/GrotschelLS81}. The ellipsoid method algorithm performs a polynomial number of iterations, therefore if the separation oracle runs in polynomial time as well, the LP is solved in overall polynomial time. Notice that the polynomial number of iterations performed by the ellipsoid method implies that the number of constraints that played a part in finding the optimum (known as \emph{active} constraints) was polynomial as well. In other words, we could effectively discard all the constraints but a polynomial number of them.
 	
 	As discussed, a separation oracle for the dual, given a proposed solution $(\vec{y};\vec{z})$, needs to find in polynomial time a violated constraint, if exists. It remains to show that such a separation oracle is polynomial-time computable.
 
 	Observe that a violated constraint to this program is a pair $i,C$ for which $C\in \mathcal{C}_i(T)$ (and therefore $\sum_{j=0}^{m-1} C_{j} \alpha_j\leq T-\sigma_i$) and at the same time $\sum_{j=0}^{m-1} C_{j}z_j> y_i$. Fortunately, for a specified $i$, finding a configuration $C$ that induces a violated constraint can be seen as finding a $k$-multiset (since $\sum_{j=0}^{m-1} C_j = k$) given by a solution to our knapsack variant:  $[m]_0$ is the item set (over which $j$ ranges), the value for the item $j$ is $z_j$, while its weight is $\alpha_j$. The given value lower bound is  $y_i$, and  $T-\sigma_i$ is the given upper bound on the weight. Effectively, we use the possibly-tighter weight bound $\min\{k\alpha_{m-1},T-\sigma_i\}$ instead, as $k\alpha_{m-1}$ bounds the overall weight obtainable with a size-$k$ multiset. As now the weight bound is polynomial in $m$ and $k$, the solution to our knapsack variant becomes polynomial.

	We repeat this knapsack-solving step for each $i$ until we find a violated constraint, or conclude that no constraint is violated.
Once we have solved the dual using the ellipsoid method with the separation oracle, we can discard all variables in the primal that do not correspond to violated constraints of the dual, since the inclusion of those constraints (resp.\ their corresponding variables) did not have any effect on the dual optimum (resp.\ the primal optimum).\footnote{In other words, the dual of the dual without the discarded constraints is the primal without their corresponding variables. Another way to explain this is that this is exactly the complementary slackness condition of the Karush-Kuhn-Tucker conditions~\cite{karush1939minima,kuhntucker}, a necessary condition for obtaining the optimum.}
The primal now contains only a polynomial number of variables and can be solved directly using the ellipsoid method or any other known polynomial solvers for LP, such as~\cite{DBLP:journals/combinatorica/Karmarkar84}.  
\end{proof}

\section{Algorithm for UCM}
Solve the above mentioned configuration-LP formulation as described in Section~\ref{sec:lp-U}. As mentioned, while both constraints are inequalities, in any solution they will actually be equalities. For each candidate $c_i$, observe the variables $x_{i,C}$, $C \in \mathcal{C}_i(T) $. Since  $\sum_{C \in \mathcal{C}_i(T)}x_{i,C} = 1$, treat the $x_{i,C}$'s as a distribution over the configurations for $i$ and randomly choose one according to that distribution. For the time being, give $c_i$ this configuration.

While every candidate now has a valid configuration (and her score does not exceed $T$), it is possible that the number of scores of a certain type is above or below $k$. Formally, if candidate $c_i$ received a configuration $C^i$, let the array $H$ such that $H[j]=\sum_{i=1}^m C^i_j$ be the \emph{histogram} of the scores. It is then possible that $H[j]\neq k$. If we would translate the configuration given to each candidate to the list of the scores awarded within it, and would write this list as the column of a matrix, this matrix might not be a relaxed manipulation matrix. In order to solve this, we need to replace some of scores in this matrix with others such that the number of scores of each type will be $k$.  On the other hand, we need to make sure this process does not add much to the score of each candidate.

Let $t=(i,j)$ be a tuple representing the event that candidate $c_i$ received a score of $\alpha_j$ in its configuration. Place all such tuples in a single multiset (if $\alpha_j$ is awarded to $i$ more than once, repeat $(i,j)$ as needed). Now sort this multiset according to the $\alpha_j$ value in an non-decreasing manner (break ties between candidates arbitrarily) thus creating the event-sequence $t_0,\ldots,t_{km-1}$, i.e., the tuples are now indexed by their rank in this sequence. We now start the actually fixing of the scores given; for each tuple $t_\ell=(i,j)$ having rank $\ell$ in the list, we change the score awarded to $c_i$ (as described by the tuple) from $\alpha_j$ to $\alpha_{\lfloor \ell/k \rfloor}$. To perform this change in the algorithm, it is enough to set $C^i_j \gets C^i_j - 1$ followed by setting $C^i_{\lfloor \ell/k \rfloor} \gets C^i_{\lfloor \ell/k \rfloor} + 1$. This is correct as $C^i_{j'}$, for any $j'$, represents the number of $\alpha_{j'}$ scores awarded to $c_i$. 

Notice that now every score is repeated $k$ times (there are only $k$ possible  $\ell$ values mapping to the same $\lfloor \ell/k \rfloor$ value). Finally, the corrected configurations represent the final strategy. This can be easily represented as a relaxed configuration matrix by referring to the matrix $[\mathrm{MS}(C^1);\cdots;\mathrm{MS}(C^m)]$, where $\mathrm{MS}(C^i)$ is a column constructed by taking the configuration $C^i$, represented as an ordered-multiset of scores (each $j$ repeats $C^i_j$ times), in some arbitrary order.

The entire process is summarized as Algorithm~\ref{alg:approx-U}.

\begin{algorithm}[t]
	\caption{Approximation algorithm.}
	\label{alg:approx-U}
	Solve the C-LP as described in Section~\ref{sec:lp-U}\;
	\lForEach{$i$}{define distribution $q$ s.t.\ $q(C)=x_{i,C}$ for all $C \in \mathcal{C}_i(T)$ and randomly choose $C^i \sim q$.}
	$L \gets \langle \rangle$\tcc*{$L$ is the empty list} 
	\ForEach{$i\in [m]$, $j \in [m]_0$ }{Append $C^i_j$ copies of $(i,j)$ to $L$ \tcc*{$j$  represents the score type $\alpha_j$}}
	Sort $L$ in an ascending order by $\scoretype(\cdot)$\tcc*{$\scoretype(t)=j$ if $t=(i,j)$}
	Re-index $L$ such that $L=\langle t_0,\ldots,t_{km-1}\rangle$\;
	\For{$\ell=0,\ldots,km-1$}{
		Observe tuple $t_\ell=(i,j)$\; 
		\tcc{Assign the score $\alpha_{\lfloor \ell/k \rfloor}$ to $c_i$, instead of the previous $\alpha_{j}$:}
		$C^i_j \gets C^i_j - 1$\;
		$C^i_{\lfloor \ell/k \rfloor} \gets C^i_{\lfloor \ell/k \rfloor} + 1$\;
	}
		
	\Return{the relaxed manipulation matrix corresponding to $C^1,\ldots,C^m$}
\end{algorithm}

Let $\beta = d \sqrt{m\log m}$ for some constant $d$. Let $g(\vec{\alpha})=\max_{i=0,\ldots,m-\beta}\alpha_{i+\beta}-\alpha_i$. In words, $g(\vec{\alpha})$ is the biggest different between a score in $\vec{\alpha}$ and another score $\beta$ entries away from it. 
\begin{lemma}
Let $C\in \{C^1,\ldots,C^m\}$ be a configuration obtained for some candidate by the rounding process, and let $C'$ be its corrected version given by the process described above. Then with arbitrary-chosen polynomially-small failure probability, \begin{equation*}
\sum_{j=0}^{m-1} \alpha_j C'_j \leq \sum_{j=0}^{m-1}  \alpha_j C_j +   k \cdot g(\vec{\alpha})\ .
\end{equation*}.
\end{lemma}  
\begin{proof}
Let $H$ be the histogram of the original configurations $C^1,\ldots,C^m$, and let the array $G$ be the array of histogram partial sums, i.e., $G[j]=\sum_{j'=0}^{j} H[j']$.  In a similar manner, define $D^i[j] = \sum_{j'=0}^{j} C^i_{j'}$ to be the partial sums array w.r.t.\ each candidate $c_i$. We will show that with high probability, $G[j] \leq (j+1)k + d k\sqrt{m \ln m}$.

Fix a specific $j$. Notice that 
\begin{equation*}
\E[G[j]] = \sum_{j'=0}^{j}\E[H[j']]=\sum_{j'=0}^{j} \sum_{\substack{i,C \\ C \in \mathcal{C}_i(T)}} C_{j'} x_{i,C} = (j+1)k
\end{equation*} 
according to the LP constraints, and  that $G[j]=\sum_{i=1}^m D^i[j]$, that is, $G[j]$ is a random variable which is the sum of the independent random variables $D^i[j]$ for $i=1,\ldots,m$. In addition, for every candidate $c_i$, it holds that $D^i[j] \in [0,k]$, as a configuration contains at most $k$ scores. Therefore, using the generalized Hoeffding inequality~\cite[Theorem~2]{hoeffding1963probability}:
\begin{align*}
\Pr \left[  G[j] - \E[G[j]] \geq \lambda \right] & \leq  \exp \left(-\frac{2\lambda^2}{\sum_{i=1}^m k^2	}\right)\\
 & =  \exp\left(-\frac{2\lambda^2 }{m k^2}\right)\ .
\end{align*} 
Setting $\lambda = d' k\sqrt{m \ln m} $, for some arbitrary constant $d'$, we get that $\Pr [  G[j] - \E[G[j]]  \geq d' k \sqrt{m\ln m} ] \leq 1/m^{d'}$, that is, the probability that we deviate from $\E[G[j]]$ by more than $\tilde{O}(k \sqrt{m})$ can be made arbitrarily polynomially small.
%
%
%
%
Using the union bound, the same can be made to hold for all $j=0,\ldots,m-1$ simultaneously.

Now observe a tuple $t_\ell=(i,j')$ before being possibly corrected by the algorithm. Since its rank $\ell$ in the sorted sequence is at most the number of scores whose type is at most $j'$, which is by definition $G[j']$, we get that $\ell \leq G[j'] \leq (j'+1)k+ d'k\sqrt{m\ln m} $, where the second inequality holds with high probability. Therefore by the algorithm changing the score $\alpha_{j'}$ to  $\alpha_{\lfloor \ell/k \rfloor} $, the score increases  by at most $\alpha_{\lfloor \ell/k \rfloor}  - \alpha_{j'}  \leq  \alpha_{(j'+1)+ d'\sqrt{m\ln m}}   -\alpha_{j'} \leq  g(\vec{\alpha}) $.

Now observe some candidate $c_i$ with a given configuration $C^i$ corrected to become a configuration $C'$ by the algorithm. Since at worst case, all of $c_i$'s $k$ scores where affected as such, her overall score has increased by at most $kg(\vec{\alpha})$.
\end{proof}

\begin{corollary}\label{corr-UCM}
The above algorithm provides an additive $kg(\vec{\alpha})$ approximation with high probability. By repeating the randomized rounding procedure a linear number of times, the failure probability becomes exponentially-small. The overall running time is polynomial.
\end{corollary}  
\begin{proof}
Let $T^*$ be the optimal value for the original problem, and let $T_{\CLP}$ be the best bound obtainable via the above C-LP combined with the binary search on $T$. Notice that $T_{\CLP} \leq T^*$, as the optimal solution is also a valid solution for the C-LP. Now observe the highest scoring candidate in the C-LP. When the algorithm terminates, we get that with high-probability her score is $T_{\CLP}+kg(\vec{\alpha}) \leq T^*+kg(\vec{\alpha})$. If we repeat the randomized rounding procedure  a linear number of times and pick the iteration yielding the minimum addition to $T_{\CLP}$, the probability of not getting a $kg(\vec{\alpha})$-approximation becomes exponentially-small.

As the additional score given by the algorithm to any other candidate is also $kg(\vec{\alpha})$, the bound $T^*+kg(\vec{\alpha})$ holds for all candidates. We conclude that this is indeed an $kg(\vec{\alpha})$ additive approximation.

As discussed, solving the C-LP is done in polynomial time (by the polynomial number of iterations of the ellipsoid method and the polynomial runtime of the $k$-multiset knapsack separation oracle). The rounding is dominated by going over a polynomial number of non-zero variables of the C-LP and is therefore polynomial as well. It is repeated a linear number fo times in order to provide an exponentially-small failure probability.
\end{proof}

From here we can directly obtain Corollary~\ref{thr:main}:
\begin{proof}[Proof of Corollary~\ref{thr:main}]
	By noticing that for Borda, $g(\vec{\alpha}) = O(\sqrt{m \log m})$.
\end{proof}

\section{Linear Programming for WCM}\label{sec:lp-W}
 When turning to the WCM problem, the `natural' LP still suffers from the deficiencies described in Section~\ref{sec:lp-U}. We again resort to using configurations. However, configurations will now be defined in a different manner, since now, when each voter has an associated weight, voters are not identical anymore and therefore our configurations need to capture the identity of the voters.
 
A configuration
$C$ for some candidate $c_i$ is now defined as a length-$k$ sequence in which $C_\ell=j$ if the voter $\ell$ awarded $\alpha_j$ to $c_i$.  
For a candidate $c_i$ and a bound $T$,  $\mathcal{C}_i(T)$ is again the set of  configurations that do not cause the candidate overall score to surpass $T$, which this time is formally   $\sum_{\ell=1}^{k}   w_\ell \alpha_{C_\ell}\leq T-\sigma_i$.
  
    The configuration LP is now formulated as follows:
    \begin{align}
    &\sum_{C \in \mathcal{C}_i(T)}x_{i,C}\leq 1 &\forall i\in[m]\ ,\label{eq:W1}\\
    &\sum_{\substack{i,C \in \mathcal{C}_i(T)\\C_\ell=j}} x_{i,C}\geq 1 &\forall j\in[m]_0, \forall \ell\in[k]  \ ,\label{eq:W2}\\
    &x_{i,C} \geq 0 & \forall i\in [m], C \in \mathcal{C}_i(T)\ .
    \end{align}     
  Again, we wish that the $x_{i,C}$'s would serve as indicator variables indicating whether or not $c_i$ was awarded with configuration $C$, (\ref{eq:W1}) guarantees that every candidate was given at most $1$ configuration and (\ref{eq:W2}) guarantees that every score was awarded by every voter at least once. The choice of inequalities over equalities will be explained soon.
   
We present another---much more complex---Knapsack variant, which will be used later by the separration orcale needed for solcing the C-LP.

 \subsection{$k$-Sequence Knapsack} Let $\{1,\ldots,m\}$ be a set of distinct items. In the \emph{$k$-sequence knapsack} problem we are required to construct a length-$k$ sequence $S=s_1,\ldots,s_k$ of items; we can repeat items from the item-set, however, we are subject to some additional constraints as will be specified immediately.
 The input to the problem is the following:
  \begin{itemize}
 	\item A value $v(j,\ell)$, for every $j \in [m]_0, \ell \in [k]$ where $v(j,\ell)$ is the value obtained by placing item $j$ at location $\ell$ in the sequence. 	
 	\item A cost $b(j)$ for each item $j$, and a penalty $p_\ell$ for each location $\ell\in [k]$. Placing an item $j$ at location $\ell$ in the sequence has a penalized-cost $p_\ell b(j)$, i.e., it depends on both the item's weight and the penalty for location $\ell$. 
 	\item A value lower-bound $V$.
 	\item A penalized-cost upper-bound $B$. 	 
\end{itemize}
	The resulting sequence $S$ should abide the following constraints:
  \begin{itemize}
 	\item  $S$'s overall value $\sum_{\ell=1}^{k}v(s_\ell,\ell)$ is  greater than $V$.
 	\item $S$'s overall penalized-cost $\sum_{\ell=1}^{k} p_\ell b(s_\ell)$ is at most  $B$.
 \end{itemize} 
%
%
%
%
%
 If such sequence $S$ exists, we should return it; otherwise we return that no such sequence exists.
 
 \begin{lemma}
 	The $k$-sequence knapsack can be solved in time polynomial in $k$,$m$, and $B$ (which is pseudo-polynomial due to the dependence on $B$).
 \end{lemma} 
 \begin{proof}
 	Similar to the proof of Lemma~\ref{lem:multiknap}, we fill out a table $Q[b',\ell]$, for $b'=0,\ldots,B$ and $\ell=0,\ldots,k$, in which $Q[b',\ell]$ is the highest value obtainable with a length-$\ell$ sequence of items of penalized-cost at most $b$. This time $Q$ is filled using a different recursion: 
 	\begin{equation}
 	Q[b',\ell]=
 	\begin{cases}
 	0 & \text{if $\ell = 0$, }\\
 	\max_j v(j,\ell)  +  Q'(b'-p_\ell b(j),\ell-1)  & \text{otherwise,}
 	\end{cases} 
 	\end{equation} 
 	where $Q'(b',\ell)=Q[b',\ell]$ if it is defined, i.e., $b'\geq 0$ and $\ell \geq 0$, and otherwise is $-\infty$.
 	
 	Therefore $Q$ can be filled-out using dynamic programming. Finally, the entry $Q[B,k]$ contains the highest value obtainable with overall cost at most $B$. Therefore, if $Q[B,k] > V$, we have found a required sequence; otherwise such does not exists. The resulting sequence itself can be recovered using backtracking on the table $Q$. Noticed that the amount of work done is $O(Bk m)$. 
 \end{proof}

\subsection{Solving the WCM C-LP}
 We return to our problem. Again, the choice of inequalities over equalities is motivated by our use of the LP dual in Theorem~\ref{thr:poly-W}. However, they have the same effect as equalities, as shown by the following lemma:
  \begin{lemma}
  	In a solution to the above C-LP, Equations~(\ref{eq:W1},\ref{eq:W2}) will actually be equalities.
  \end{lemma}
  \begin{proof}
  	Notice that by Equation~(\ref{eq:W2}):
  	\begin{align}
  	km &\leq \sum_{j,\ell} \sum_{i=1}^m \sum_{\substack{C \in \mathcal{C}_i(T)\\C_\ell=j}}  x_{i,C} \\
  	&=\sum_{i=1}^m \sum_{C \in \mathcal{C}_i(T) } \sum_{\substack{\ell,j\\C_\ell=j}} x_{i,C} \\
  	&=\sum_{i=1}^m \sum_{C \in \mathcal{C}_i(T)} x_{i,C} \sum_{\substack{\ell,j\\C_\ell=j}}1 \\
  	&=k\sum_{i=1}^m \sum_{C \in \mathcal{C}_i(T)} x_{i,C}\label{eq:W3}\\
  	&\leq km \label{eq:W4}
  	\end{align}
  	where (\ref{eq:W4}) holds by plugging (\ref{eq:W1}) into (\ref{eq:W3}). We therefore obtain that $$\sum_{j,\ell} \sum_{i=1}^m \sum_{C \in \mathcal{C}_i(T)} x_{i,C}=km$$ which forces both above non-trivial LP inequalities to be equalities.
  \end{proof}
 
 \begin{theorem}\label{thr:poly-W}
 	Given a fixed value $T$, the WCM C-LP can be solved in polynomial time.
 \end{theorem}
 \begin{proof}
 	We again refer to the LP-dual, which this time is:
 	\begin{align*}
 	&\min_{\vec{y},\vec{z}} \sum_{i=1}^m y_i - \sum_{j,\ell} z_{j,\ell} \\
 	\mbox{subject to:}\\
 	\sum_{\substack{j,\ell\\C_\ell=j}} z_{j,\ell}&\leq y_i && \forall i\in [m], C \in \mathcal{C}_i(T) \\
 	y_i &\geq 0 && \forall i=1,\ldots,m\\
 	z_{j,\ell} &\geq 0 && \forall (j,\ell) \in [m]_0\times[k]
 	\end{align*}
 	Furthermore, the above single non-trivial constraint can be more conveniently re-written as
 	\begin{align*}
 	\sum_{\ell=1}^{k} z_{C_\ell,\ell}&\leq y_i && \forall i\in [m], C \in \mathcal{C}_i(T) \ .
 	\end{align*}
 	We again turn to the ellipsoid method with a separation oracle; this time, a violated constraint to this program is a pair $i,C$ for which $C\in \mathcal{C}_i(T)$ (and therefore $\sum_{\ell=1}^{k}   w_\ell \alpha_{C_\ell}\leq T-\sigma_i$) and at the same time $\sum_{\ell=1}^{k} z_{C_\ell,\ell} > y_i$. For a specified $i$, finding a configuration $C$ that induces a violated constraint can be seen as finding a $k$-sequence given by a solution to our knapsack variant:  $[m]_0$ is the item set (over which $j$ ranges), the value for placing item $j$ at location $\ell$ is $z_{j,\ell}$, item $j$'s cost is $\alpha_j$, and the penalty for location $\ell$ is $w_\ell$. The given value lower bound is  $y_i$, and  $T-\sigma_i$ is the given upper bound on the penalized cost. Effectively, we use the possibly-tighter weight bound $\min\{W\alpha_{m-1},T-\sigma_i\}$ instead, as $W\alpha_{m-1}$ bounds the overall cost obtainable with a length-$k$ sequence. As now the weight bound is polynomial in $\alpha_{m-1}$ and $W$, the solution to our knapsack variant becomes polynomial.

	We repeat this knapsack-solving step for each $i$ until we find a violated constraint, or conclude that no constraint is violated.
Once we have solved the dual using the ellipsoid method with the separation oracle, we continue in a similar fashion to the proof of Theorem~\ref{thr:poly-U}.
\end{proof}

\section{Algorithm for WCM}
Solve the above mentioned configuration-LP formulation as described in Section~\ref{sec:lp-W}. As both constraints will be equalities, for each candidate $c_i$, we treat the $x_{i,C}$'s as a distribution over the configurations for $i$ since $\sum_{C \in \mathcal{C}_i(T)}x_{i,C} = 1$,  and randomly choose one according to that distribution. For the time being, give $i$ this configuration.

As for UCM, every candidate now has a valid configuration but constraints may still be violated; it is possible that the number of scores of a certain type $\alpha_j$ given by a specific voter $\ell$ is not exactly $1$. Formally, fix a specific voter $\ell$; we let the array $H$ such that $H[j]=\abs{\{i \mid C^i_\ell =j\} }$ be histogram of the scores with respect to $\ell$. It is then possible that $H[j]\neq 1$. Our goal, as before, is to fix this without introducing too much of an addition to the candidates' overall scores. However, there is now some added complexity due to the necessity to preserve the identity of the voter when fixing a specific score given.

Let $t=(i,\ell,j)$ be a tuple representing the event that candidate $c_i$ received a score of $\alpha_j$ from voter $\ell$ in its configuration. Fix a manipulator $\ell$, and place all tuples having $\ell$ as their respective voter in a set $A$, that is $A=\{ (i,\ell,C^i_\ell) \mid i = 1,\ldots,m \}$. Sort $A$ according to each tuple $(i,\ell,C^i_\ell)$'s score-index $C^i_\ell$, and let $L=\langle t_1,\ldots , t_m \rangle$ be the resulting list. Notice that  any tuple $t_j=(i,\ell,j')$ in $L$ represents the event that currently  $C^i_\ell=j'$. Now change $C^i_\ell$ such that $C^i_\ell \gets j$. In words, $C^i_\ell$ gets the rank of its respective tuple in the sorted list $L$. In effect, the score awarded to $c_i$ by $\ell$  changes from $\alpha_{j'}$ to $\alpha_j$.



We repeat the above process for each voter. Notice that now every score is repeated $k$ times, one by each voter.
The process is summarized as Algorithm~\ref{alg:approx-W}.

\begin{algorithm}[t]
	\caption{WCM Approximation algorithm.}
	\label{alg:approx-W}
	Solve the C-LP as described in Section~\ref{sec:lp-W}\;
	\lForEach{$i$}{define distribution $q$ s.t.\ $q(C)=x_{i,C}$ for all $C \in \mathcal{C}_i(T)$ and randomly choose $C^i \sim q$.}
	\For{$\ell\gets 1$ \KwTo $k$}{
		Let $A=\{ (i,\ell,C^i_\ell) \mid i = 1,\ldots,m \}$\tcc*{tuple $(i,\ell,j)$ represnts the event the $\ell$ awarded $\alpha_j$ to $c_i$}
		Sort $A$ in an ascending order by $\scoretype(\cdot)$ and let $L=\langle t_0,\ldots, t_{m-1} \rangle $ be the resulting list.\tcc*{$\scoretype(t)=j$ if $t=(i,\ell,j)$}
		\For{$j\gets 0$ \KwTo $m-1$}{
			Observe tuple $t_j=(i,\ell,j')$.\;
			$C^i_\ell \gets j$
			\tcc*{this assigns the score $\alpha_j$ to $c_i$, instead of the previous $\alpha_{j'}$}}
	}
	\Return{the resulting manipulation matrix $[C^1;\cdots;C^m]$}\tcc*{place $C^i$ as the $i$-th column vector of the resulting matrix}
\end{algorithm}

As before, let $\beta = d \sqrt{m\log m}$ for some constant $d$ and  let $g(\vec{\alpha})=\max_{i=0,\ldots,m-\beta}\alpha_{i+\beta}-\alpha_i$. 
\begin{lemma}
Let $C\in \{C^1,\ldots,C^m\}$ be a configuration obtained for some candidate by the rounding process, and let $C'$ be its corrected version given by the process described above. Then with arbitrary-chosen polynomially-small failure probability, \begin{equation*}
\sum_{j=0}^{m-1} \alpha_j C'_j \leq \sum_{j=0}^{m-1}  \alpha_j C_j +   W \cdot g(\vec{\alpha})\ .
\end{equation*}.
\end{lemma}  
\begin{proof}
Fix a specific voter $\ell$. Let $H$ be the histogram with respect to voter $\ell$ over the original configurations $C^1,\ldots,C^m$, i.e. $H[j]=\abs{\{i \mid C^i_\ell =j\} }$. Let the array $G$ be the array of histogram partial sums, i.e., $G[j]=\sum_{j'=0}^{j} H[j']=\abs{\{i \mid C^i_\ell \leq j\} }$.  We also define $D^i[j]$ to be a Bernoulli variable which is equal to $1$ if $C^i_\ell \leq j$, and $0$ otherwise. We will show that with high probability, $G[j] \leq (j+1) + d \sqrt{m \ln m}$.

Fix a specific $j$. Notice that 
\begin{equation*}
\E[G[j]] = \sum_{j'=0}^{j}\E[H[j']]=\sum_{j'=0}^{j} \sum_{\substack{i, C\\C \in \mathcal{C}_i(T)\\C_{\ell}=j'}}  x_{i,C} = j+1
\end{equation*} 
according to the LP constraints, and  that $G[j]=\sum_{i=1}^m D^i[j]$, that is, $G[j]$ is a random variable which is the sum of the independent random variables $D^i[j]\in \{0,1\}$ for $i=1,\ldots,m$. Therefore, using the `classic' Hoeffding inequality:
\begin{equation*}
\Pr \left[  G[j] - \E[G[j]] \geq \lambda \right]  \leq e^{-\frac{2 \lambda^2}{m}}\ .
\end{equation*} 
Setting $\lambda = d' \sqrt{m\ln m} $, for some arbitrary constant $d'$, we get that $\Pr [  G[j] - \E[G[j]]  \geq d' \sqrt{m\ln m} ] \leq 1/m^{d'}$, that is, the probability that we deviate from $\E[G[j]]$ by more than $\tilde{O}(\sqrt{m})$ can be made arbitrarily polynomially small.
%
%
%
%
Using the union bound, the same can be made to hold for all $j=0,\ldots,m-1$ and $\ell=1,\ldots,k$ simultaneously.

Now observe a tuple $t_j=(i,\ell, j')$ before being possibly corrected by the algorithm. Since its rank $j$ in the sorted sequence is at most the number of scores given by $\ell$ whose type is at most $j'$, which is by definition $G[j']$, we get that $j \leq G[j'] \leq (j'+1)+ d'\sqrt{m\ln m} $, where the second inequality holds with high probability. Therefore by the algorithm changing the score $\alpha_{j'}$ to  $\alpha_{j} $, the score increases  by at most $\alpha_{j}  - \alpha_{j'}  \leq  \alpha_{(j'+1)+ d'\sqrt{m\ln m}}   -\alpha_{j'} \leq  g(\vec{\alpha}) $.

Now observe some candidate $c_i$ with a given configuration $C^i$ corrected to become a configuration $C'$ by the algorithm. Since at worst case, all of $c_i$'s $k$ scores where affected as such, her overall score has increased by at most $\sum_{\ell=1}^k (w_\ell g(\vec{\alpha})) = Wg(\vec{\alpha}) $.
\end{proof}

\begin{corollary}
The above algorithm provides an additive $Wg(\vec{\alpha})$ approximation with high probability. By repeating the randomized rounding procedure a linear number of times, the failure probability becomes exponentially-small. The overall running time is polynomial.
\end{corollary}  
\begin{proof}
Identical to the proof of Corollary~\ref{corr-UCM}, the only difference being $W$ used instead of $k$.
%
%
\end{proof}

From here we can directly obtain Corollary~\ref{thr:main2}:
\begin{proof}[Proof of Corollary~\ref{thr:main2}]
	By noticing that for Borda, $g(\vec{\alpha}) = O(\sqrt{m \log m})$.
\end{proof}

\section{Implementation}
We implemented our algorithm for the case for Borda-UCM and Borda-WCM and uploaded the code to a public repository\footnote{\texttt{https://github.com/okeller/BordaManipulation}}. The main subroutine is solving the LP duals that we have defined. However, the dependency on an LP solver using the ellipsoid method with a separation oracle proved difficult as to the best of our knowledge there is no library which enables solving an LP this way. Instead we simulated this by using a  general LP-solving library~\cite{cvxopt,glpk} and running the separation oracle externally as described in Algorithm~\ref{alg:simul}.
\begin{algorithm}[t]
	\caption{Simulating the ellipsoid method with a separation oracle.}
	\label{alg:simul}
	\KwIn{A linear program $P=(f,S)$ with an objective function $f$ and a separation oracle $S$ (instead of an explicit list of constraits)}
	Let $R\gets \emptyset$\tcc*{$R$ will be a list of effective constraints}
	Let $P' \gets (f,R)$ \tcc*{$P'$ is $P$ without any constraints}
	$\vec{x} \gets \LPS(P')$\;	
	\While{$S(\vec{x})$ returns a violated constraint $r$}{	 
		$R\gets R \cup \{r\}$\;
		$P' \gets (f,R)$\;
		$\vec{x} \gets \LPS(P')$\;
	}
	\Return{$\vec{x}$}
\end{algorithm}

\subsection{Practical Heuristics}
We used the following heuristics in our implementation:
\begin{itemize}
	\item We obtained a running-time speedup was modifying the separation oracle to return a set of violated constraints, one for each $i$ (if such exists), instead of a single violated constraint, and adding all of them to the list of constraints. 
	\item When sorting by the score index after the rounding part, we broke ties between tuples having the same score in favor of the candidate with the higher current score. This way, candidates with lower score are more likely to have a score index awarded to them increased by the fixing procedure.   
	\item   Our practical running time is dominated by solving the C-LP. Therefore, we can perform the randomized part of rounding process several times, and pick the best one, while incurring only a negligible overhead to the runtime. 
\end{itemize}

 
 \section{Experiments}
 \subsection{Experimental Setting}
 To evaluate our algorithms for both UCM and WCM in a well-studied setting, we have run experiments for Borda-UCM and Borda-WCM on sets of values for $n$, $k$ and $m$ (choice of values will be explained shortly). For the unweighted case, our algorithm was compared to  \textsc{Average Fit}  that was shown empirically to  outperform \rev{}~\cite{DBLP:journals/ai/DaviesKNWX14}. For the weighted case, as we are not aware of a generalization of \textsc{Average Fit} to a weighted setting, we compared against \rev{}.

 We have chosen $n=2k$, as having one manipulator for every two original voters represents a sweet-spot where manipulators have enough power to change outcomes, but not too much; for each $n,k,m$ combination, we have run $20$ experiments in which we have drawn the non-manipulator votes from a uniform distribution. We have chosen  $k \approx \sqrt{m}$ or smaller, as our algorithms are suited for low $k$ values:  as mentioned, our algorithm for Borda-UCM is theoretically competitive when $k=o(\sqrt{m / \log m})$, and we have wished to verify this also in an empirical setting. Lower $k$ values are also the cases which are more difficult to \textsc{Average Fit} heuristic of~\cite{DBLP:journals/ai/DaviesKNWX14}.

 \subsection{Borda-UCM}
 We are comparing our results to those obtained by \textsc{Average Fit}, and to the fractional solution, $T_{\CLP}$. The results are summarized in Figure~\ref{fig:res}.
 As can be seen, our algorithm performs well in practice.
In the many of  the cases, both the C-LP and \textsc{Average Fit} were known to be identical to the optimal solution, as depicted in Figure~\ref{subfig:res_a}:  these are the instances in which the C-LP solution  did not increase when performing the LP rounding. Therefore, in those cases we know what is  the real optimal solution, as $T_{\CLP}$ (resp.\ our final algorithm result) provides a lower (resp.\ an upper) bound for it, and both are equal. Thus our formulation in many times provides a very efficient method for verifying whether our algorithm (or any other algorithm) finds an optimal solution.   
 
In cases where the C-LP solution increased when performing the LP rounding, C-LP tends to be better than \textsc{Average Fit} on small $k$ values, and worse on larger ones. This is depicted in Figures~\ref{subfig:res_b} and~\ref{subfig:res_c}. For instance, consider our example from before, this time with $p$: $k=2$, $m=5$ and  $(\sigma_0,\ldots,\sigma_5)=(0,5,6,6,6,7)$. Obviously both methods will award $p$ with $5+5=10$. However in ours the top score of a candidate who is not $p$ will be $10$, and in \textsc{Average Fit} it is $12$. Therefore, the choice of algorithm determines if $p$ wins or not by a difference of $2$.
 \begin{figure}[hbpt]
 	\centering
 	\begin{subfigure}[hbpt]{\columnwidth}
 	        \includegraphics[scale=0.46]{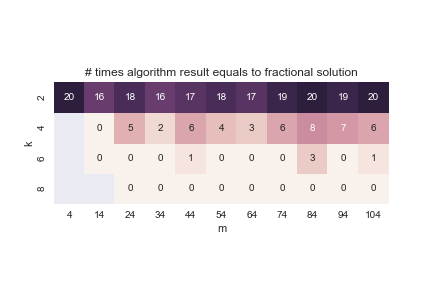}
 	        \caption{Number of cases for various value of $m$ and $k$ (out of $20$ experiments carried out for each $m,k$), where our algorithm provided results equal to the fractional solution $T_{\CLP}$ (and therefore, to the optimum $T^*$).}
 	        \label{subfig:res_a}
 	\end{subfigure}
 	\begin{subfigure}[hbpt]{\columnwidth}
 	        \includegraphics[scale=0.46]{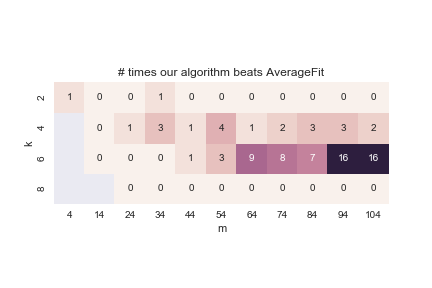}
 	        \caption{Number of cases for various value of $m$ and $k$ (out of $20$ experiments carried out for each $m,k$), where the C-LP method provided strictly better results compared to \textsc{Average Fit}.}
 	        \label{subfig:res_b}
 	\end{subfigure}
	 \begin{subfigure}[hbpt]{\columnwidth}
	 	\includegraphics[scale=0.46]{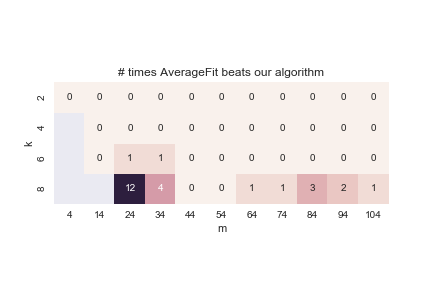}
	 	\caption{Number of cases for various value of $m$ and $k$ (out of $20$ experiments carried out for each $m,k$), where the \textsc{Average Fit} provided strictly better results compared to the C-LP method.}
	 	\label{subfig:res_c}
	 \end{subfigure}
 	\caption{Experimental results for Borda-UCM.}
 	\label{fig:res}
 \end{figure}

 \subsection{Borda-WCM}
 For the weighted case, as \textsc{Average Fit} has no weighted generalization, our competitor is the weighted variant of \rev{}. Here the results were much more conclusive: we experimented with the same type of values, choosing the manipulator weights randomly from $\{1,2\}$. The reasoning behind this was to explore the weighted case on one hand, but on the other hand not to give any one manipulator excessive power over others by giving him   a large weight. As the similar tables show (Figure~\ref{fig:res-w}), our algorithm beats \rev{} on any instance.
\begin{figure}[hbpt]
	\centering
	\begin{subfigure}[hbpt]{\columnwidth}
		\includegraphics[scale=0.46]{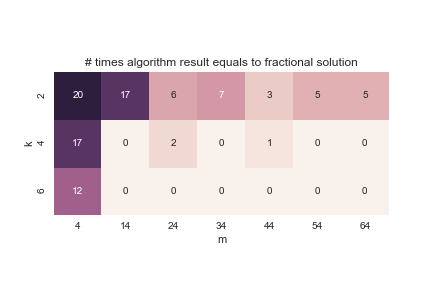}
		\caption{Number of cases for various value of $m$ and $k$ (out of $20$ experiments carried out for each $m,k$), where our algorithm provided results equal to the fractional solution $T_{\CLP}$ (and therefore, to the optimum $T^*$).}
		\label{subfig:res_a-w}
	\end{subfigure}
	\begin{subfigure}[hbpt]{\columnwidth}
		\includegraphics[scale=0.46]{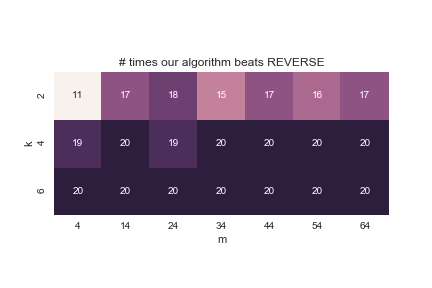}
		\caption{Number of cases for various value of $m$ and $k$ (out of $20$ experiments carried out for each $m,k$), where the C-LP method provided strictly better results compared to \rev{}.}
		\label{subfig:res_b-w}
	\end{subfigure}
	\begin{subfigure}[hbpt]{\columnwidth}
		\includegraphics[scale=0.46]{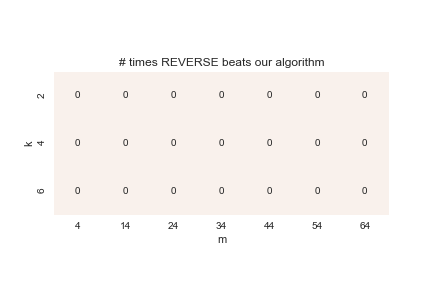}
		\caption{Number of cases for various value of $m$ and $k$ (out of $20$ experiments carried out for each $m,k$), where  \rev{} provided strictly better results compared to the C-LP method.}
		\label{subfig:res_c-w}
	\end{subfigure}
	\caption{Experimental results for Borda-WCM.}
	\label{fig:res-w}
\end{figure}

\section{Conclusions}
We have presented  additive approximations to the score of the highest-scoring non-preferred candidate in general scoring rules. It enables us to find a winning strategy for $p$ in cases where other methods would not necessarily have found one. Our method employs new techniques based on configuration linear programs. There are several interesting directions for future work, which vary from the concrete to the general:
\begin{enumerate}
  \item Understand in which instances different manipulation strategies outperform others.
  \item Find algorithms that can guarantee victory even if the margin is smaller, or prove $\NP$-hardness even if there is a solution with margin $O(\sqrt{m})$.
  \item Apply a C-LP to other voting methods.
  \item Apply  a C-LP to other problems in computational social choice, such as fair division of multiple indivisible goods.
\end{enumerate}

\section*{Acknowledgments}{
We thank Sarit Kraus and Ariel Procaccia for insightful discussions.

This work was supported by the Israel Science Foundation, under Grant No. 1488/14 and Grant No. 1394/16.
}

\bibliography{voting}
\bibliographystyle{abbrv}

\end{document}